\documentclass{sig-alternate-2013}

\usepackage{graphicx}
\usepackage{amsmath,amssymb,mathtools}
\usepackage{paralist}
\usepackage{bm}
\usepackage{xspace}
\usepackage{url}
\usepackage{boxedminipage}
\usepackage{wrapfig}
\usepackage{ifthen}
\usepackage{color}
\usepackage{xcolor}
\usepackage{framed}
\usepackage{algorithmic,algorithm}
\usepackage{subfig}
\usepackage{graphicx}

\usepackage[pagebackref,letterpaper=true,colorlinks=true,pdfpagemode=none,urlcolor=blue,linkcolor=blue,citecolor=violet,pdfstartview=FitH]{hyperref}
\newtheorem{theorem}{Theorem}[section]

\newtheorem{lemma}[theorem]{Lemma}

\newtheorem{problem}[theorem]{Problem}

\newcommand{\ignore}[1]{}

\newcommand{\cA}{{\cal A}}
\newcommand{\cB}{\mathcal{B}}

\newcommand{\cE}{{\cal E}}

\newcommand{\cN}{{\cal N}}

\newcommand{\sgn}{\mathrm{sgn}}

\newcommand{\RR}{\mathbb{R}}

\newcommand{\EX}{\hbox{\bf E}}

\newcommand{\whimp}{WHIMP}

\newcommand{\est}{\mathrm{est}}
\newcommand{\nnz}{\mathrm{nnz}}

\newcommand{\Sec}[1]{\hyperref[sec:#1]{\S\ref*{sec:#1}}} %section
\newcommand{\Eqn}[1]{\hyperref[eq:#1]{(\ref*{eq:#1})}} %equation
\newcommand{\Fig}[1]{\hyperref[fig:#1]{Fig.\,\ref*{fig:#1}}} %figure
\newcommand{\Tab}[1]{\hyperref[tab:#1]{Tab.\,\ref*{tab:#1}}} %table
\newcommand{\Thm}[1]{\hyperref[thm:#1]{Theorem\,\ref*{thm:#1}}} %theorem
\newcommand{\Fact}[1]{\hyperref[fact:#1]{Fact\,\ref*{fact:#1}}} %fact
\newcommand{\Lem}[1]{\hyperref[lem:#1]{Lemma\,\ref*{lem:#1}}} %lemma
\newcommand{\Prop}[1]{\hyperref[prop:#1]{Prop.~\ref*{prop:#1}}} %property
\newcommand{\Cor}[1]{\hyperref[cor:#1]{Corollary~\ref*{cor:#1}}} %corollary
\newcommand{\Conj}[1]{\hyperref[conj:#1]{Conjecture~\ref*{conj:#1}}} %conjecture
\newcommand{\Def}[1]{\hyperref[def:#1]{Definition~\ref*{def:#1}}} %definition
\newcommand{\Alg}[1]{\hyperref[alg:#1]{Alg.~\ref*{alg:#1}}} %algorithm
\newcommand{\Ex}[1]{\hyperref[ex:#1]{Ex.~\ref*{ex:#1}}} %example
\newcommand{\Clm}[1]{\hyperref[clm:#1]{Claim~\ref*{clm:#1}}} %example
\newcommand{\Prb}[1]{\hyperref[prb:#1]{Problem~\ref*{prb:#1}}} %problem
\newcommand{\Step}[1]{\hyperref[step:#1]{Step~\ref*{step:#1}}} %problem

\graphicspath{{Images/}}

\permission{\copyright 2017 International World Wide Web Conference Committee \\ (IW3C2), published under Creative Commons CC BY 4.0 License.}
\conferenceinfo{WWW 2017,}{April 3--7, 2017, Perth, Australia.}
\copyrightetc{ACM \the\acmcopyr}
\crdata{978-1-4503-4913-0/17/04. \\
http://dx.doi.org/10.1145/3038912.3052633 \\
\includegraphics{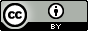}
}
\begin{document}

%\setcopyright{acmcopyright}
%\setcopyright{acmlicensed}
%\setcopyright{rightsretained}
%\setcopyright{usgov}
%\setcopyright{usgovmixed}
%\setcopyright{cagov}
%\setcopyright{cagovmixed}

% DOI
% \doi{10.475/123_4}

% ISBN
% \isbn{123-4567-24-567/08/06}

%Conference
% \conferenceinfo{PLDI '13}{June 16--19, 2013, Seattle, WA, USA}

% \acmPrice{\$15.00}

%
% --- Author Metadata here ---
% \conferenceinfo{WOODSTOCK}{'97 El Paso, Texas USA}
%\CopyrightYear{2007} % Allows default copyright year (20XX) to be over-ridden - IF NEED BE.
%\crdata{0-12345-67-8/90/01}  % Allows default copyright data (0-89791-88-6/97/05) to be over-ridden - IF NEED BE.
% --- End of Author Metadata ---

\title{When Hashes Met Wedges: A Distributed Algorithm
for Finding High Similarity Vectors}

% Just remember to make sure that the TOTAL number of authors
% is the number that will appear on the first page PLUS the
% number that will appear in the \additionalauthors section.

\numberofauthors{3}

\author{
\alignauthor
Aneesh Sharma\\
\affaddr{Twitter, Inc.}\\
\email{aneesh@twitter.com}
\alignauthor
C. Seshadhri\\
\affaddr{University of California}\\
\affaddr{Santa Cruz, CA}\\
\email{sesh@ucsc.edu}
\alignauthor
Ashish Goel \thanks{Research supported in part by NSF Award 1447697.}\\
\affaddr{Stanford University}\\
\email{ashishg@stanford.edu}
}
\maketitle

\begin{abstract}
  Finding similar user pairs is a fundamental task in social networks,
  with numerous applications in ranking and personalization tasks such
  as link prediction and tie strength detection. A common
  manifestation of user similarity is based upon network structure:
  each user is represented by a vector that represents the user's
  network connections, where pairwise cosine similarity among these
  vectors defines user similarity. The predominant task for user
  similarity applications is to discover all similar pairs that have a
  pairwise cosine similarity value larger than a given threshold
  $\tau$. In contrast to previous work where $\tau$ is assumed to be
  quite close to 1, we focus on recommendation applications where
  $\tau$ is small, but still meaningful. The all pairs cosine
  similarity problem is computationally challenging on networks with
  billions of edges, and especially so for settings with small
  $\tau$. To the best of our knowledge, there is no practical solution
  for computing all user pairs with, say $\tau = 0.2$ on large social
  networks, even using the power of distributed algorithms.

  Our work directly addresses this challenge by introducing a new
  algorithm --- WHIMP{} --- that solves this problem efficiently in
  the MapReduce model.  The key insight in WHIMP{} is to combine the
  ``wedge-sampling" approach of Cohen-Lewis for approximate matrix
  multiplication with the SimHash random projection techniques of
  Charikar. We provide a theoretical analysis of WHIMP, proving that
  it has near optimal communication costs while maintaining
  computation cost comparable with the state of the art.  We also
  empirically demonstrate WHIMP{}'s scalability by computing all highly
  similar pairs on four massive data sets, and show that it accurately
  finds high similarity pairs. In particular, we note that WHIMP{}
  successfully processes the entire Twitter network, which has tens of
  billions of edges.

%
% More generally, we focus on the following problem.
% Given two massive non-negative matrices $A, B$ and a threshold $\tau$,
% we wish to determine all entries in $AB$ that are at least $\tau$. How can we
% perform this task \emph{without} performing the potentially expensive full matrix
% multiplication?
%
% An important application is in similarity search in a large set of vectors, and the sizes
% of $A$ and $B$ can be around 100 billion non-zeroes in industrial applications.
% To the best of our knowledge, there is no solution that scales to these sizes, primarily
% due to communication/shuffle costs.
%
%
% We give a full theoretical analysis and prove that the
% shuffle cost of \whimp{} is near optimal.
% We demonstrate \whimp{} on massive data sets and show that it accurately finds
% all large entries in actual matrix products.
%
\end{abstract}


\section{Introduction} \label{sec:intro}

Similarity search among a collection of objects is one of the oldest
and most fundamental operations in social networks, web mining, data
analysis and machine learning. It is hard to overstate the importance
of this problem: it is a basic building block of personalization and
recommendation systems~\cite{DDGR07, GGLS+13}, link
prediction~\cite{AdAd03,LiKl07}, and is found to be immensely useful
in many personalization and mining tasks on social networks and
databases~\cite{XNR10,MSLC01,YYH05, RRS11book}. Indeed, the list of
applications is so broad that we do not attempt to survey them here
and instead refer to recommender systems and data mining textbooks
that cover applications in diverse areas such as collaborative
filtering~\cite{RRS11book, mmds_book}.

Given the vast amount of literature on similarity search, many forms
of the problem have been studied in various applications. In this work
we focus on the social and information networks setting where we can
define pairwise similarity among users on the network based on having
common connections. This definition of similarity is particularly
relevant in the context of information networks where users generate
and consume content (Twitter, blogging networks, web networks,
etc.). In particular, the directionality of these information networks
provides a natural measure that is sometimes called ``production''
similarity: two users are defined to be similar to each other if they
are followed by a common set of users. Thus, ``closeness" is based on
common followers, indicating that users who consume content from one
of these users may be interested in the other ``producer'' as
well.\footnote{One can also define ``consumption'' similarity, where
  users are similar if they follow the same set of users.} The most
common measure of closeness or similarity here is \emph{cosine
  similarity}. This notion of cosine similarity is widely used for
applications~\cite{AdAd03,KaMeCh05,AnPi05,LiKl07} and is in particular
a fundamental component of the Who To Follow recommendation system at
Twitter~\cite{GSWY13,GGLS+13}.

Our focus in this work is on the computational aspect of this widely
important and well studied problem. In particular, despite the large
amount of attention given to the problem, there remain significant
scalability challenges with computing {\em all-pairs} similarity on
massive size information networks. A unique aspect of this problem on
these large networks is that cosine similarity values that are
traditionally considered ``small'' can be quite meaningful for social
and information network applications --- it may be quite useful and
indicative to find users sharing a cosine similarity value of 0.2, as
we will illustrate in our experimental results. With this particular
note in mind, we move on to describing our problem formally and
discuss the challenges involved in solving it at scale.

\begin{figure*}[t]
  \centering
  \subfloat{\label{fig:pr30}
    % TRIM LEFT BOTTOM RIGHT TOP
    \includegraphics[width=.3\textwidth,trim=0 0 0 0]{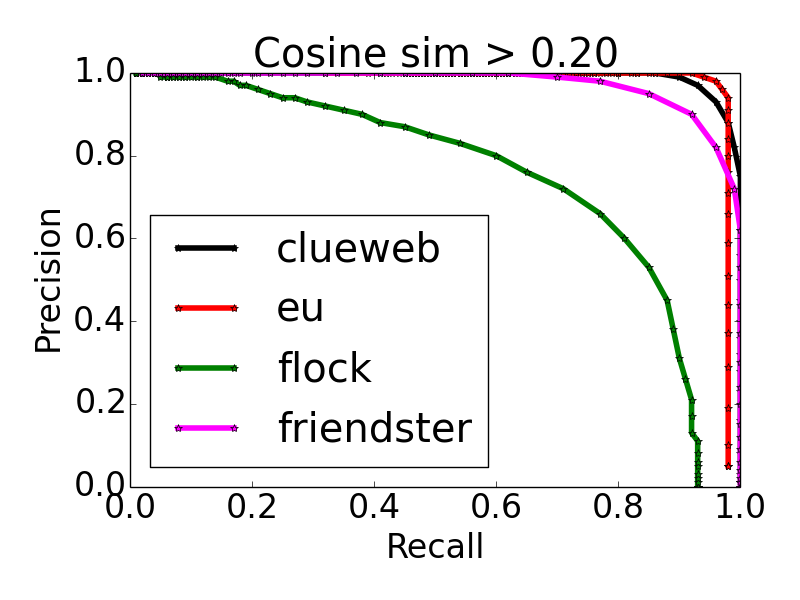}
  }
  \subfloat{\label{fig:eu30}
    % TRIM LEFT BOTTOM RIGHT TOP
    \includegraphics[width=.3\textwidth,trim=0 0 0 0]{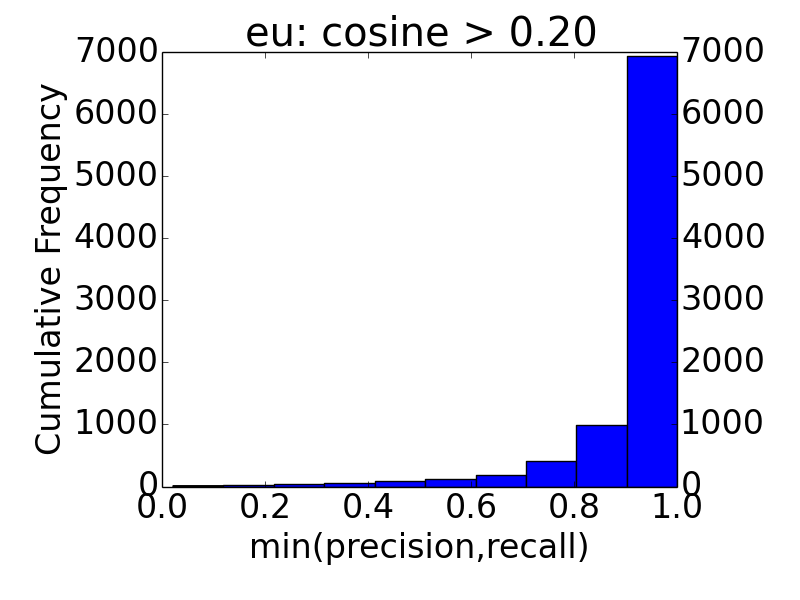}
  }
  \subfloat{\label{fig:flock30}
    % TRIM LEFT BOTTOM RIGHT TOP
    \includegraphics[width=.3\textwidth,trim=0 0 0 0]{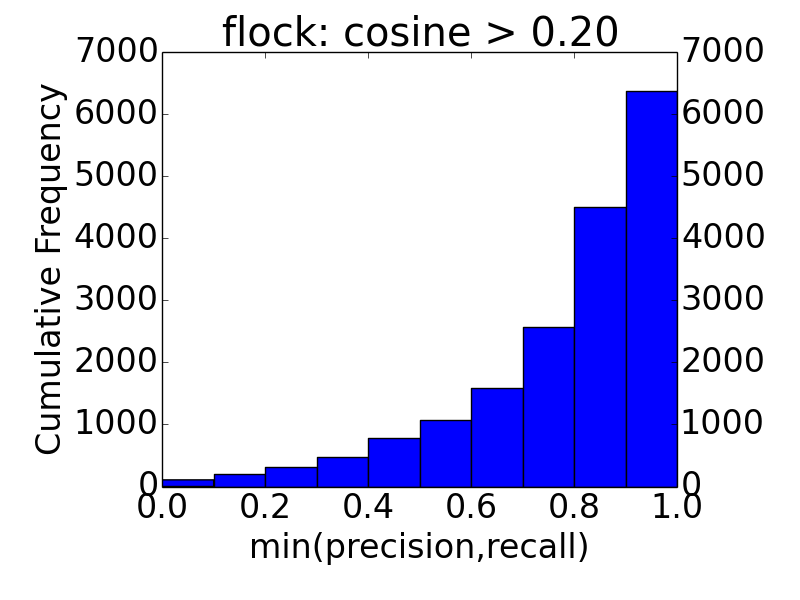}
  }
  \caption{Result on \whimp{} for $\tau = 0.2$: the left plot the precision-recall curves
  for finding all entries in $A^TA$ above $0.2$ (with respect to a sampled evaluation set).
  The other plots give the cumulative distribution, over all sampled users,
  of the minimum of precision and recall. We observe that for an overwhelming majority of users,
  \whimp{} reliably finds more than 70\% of $0.2$-similar users.
  }
  \label{fig:intro}
\end{figure*}

\subsection{Problem Statement} \label{sec:problem}

As mentioned earlier, the similarity search problem is relevant to a
wide variety of areas, and hence there are several languages for
describing similarity: on sets, on a graph, and also on matrix
columns. We'll attempt to provide the different views where possible,
but we largely stick to the matrix notation in our work. Given two
sets $S$ and $T$, their cosine similarity is
$|S\cap T| / \sqrt{|S|\cdot|T|}$, which is a normalized intersection
size. It is instructive to define this geometrically, by representing
a set as an incidence vector.  Given two (typically non-negative)
vectors $\vec{v}_1$ and $\vec{v}_2$, the cosine similarity is
$\vec{v}_1 \cdot \vec{v}_2/(\|\vec{v}_1\|_2 \|\vec{v}_2\|_2)$.  This
is the cosine of the angle between the vectors; hence the name.

In our context, the corresponding $\vec{v}$ for some user is the
incidence vector of followers. In other words, the $i$th coordinate
of $\vec{v}$ is $1$ if user $i$ follows the user, and $0$ otherwise.
Abusing notation, let us denote the users by their corresponding vectors,
and we use the terms ``user" and ``vector" interchangably. Thus, we can define our problem as follows.

\begin{problem} \label{prb:hemp}
Given a sets of vectors $\vec{v}_1, \vec{v}_2, \ldots, \vec{v}_m$ in $(\RR^+)^d$,
and threshold $\tau > 0$:
determine all pairs $(i,j)$ such that $\vec{v}_i\cdot \vec{v}_j \geq \tau$.

Equivalently, call $\vec{v}_j$ $\tau$-similar to $\vec{v}_i$ if $\vec{v}_i \cdot \vec{v}_j \geq \tau$.
For every vector $\vec{v}_i$, find all vectors $\tau$-similar to $\vec{v}_j$.
\end{problem}

In terms of (approximate) information retrieval, the latter formulation represents a more
stringent criterion. Instead of good accuracy in find similar pairs overall, we demand
high accuracy for most (if not all) users. This is crucial for any recommendation
system, since we need good results for most users. More generally, we want good
results at all ``scales", meaning accurate results for users with small followings
as well as big followings.  Observe
that the sparsity of $\vec{v}$ is inversely related to the indegree (following size)
of the user, and represents their popularity.
Recommendation needs to be of high quality both for
newer users (high sparsity $\vec{v}$) and celebrities (low sparsity $\vec{v}$).

We can mathematically express \Prb{hemp} in matrix terms as follows.
Let $A$ be the $d \times n$ matrix where the $i$th column is
$\vec{v}_i/\|\vec{v}_i\|_2$.  We wish to find all large entries in the
Gramian matrix $A^TA$ (the matrix of all similarities). It is
convenient to think of the input as $A$. Note that the non-zeros of
$A$ correspond exactly to the underlying social network edges.

\subsection{Challenges} \label{sec:chal}

\textbf{Scale:} The most obvious challenge for practical applications
is the sheer size of the matrix $A$.  For example, the Twitter
recommendation systems deal with a matrix with hundreds of millions of
dimensions, and the number of non-zeros is in many tens of
billions. Partitioning techniques become extremely challenging for
these sizes and clearly we need distributed algorithms for \Prb{hemp}.

\textbf{The similarity value $\tau$:} An equally important (but less
discussed) problem is the relevant setting of threshold $\tau$ in \Prb{hemp}.
In large similarity search applications, a cosine similarity (between users)
of, say, $0.2$ is highly significant. Roughly speaking, if user
$u$ is $0.2$-similar to $v$, then 20\% of $u$'s followers also follow $v$.
For recommendation, this is an immensely strong signal. But for
many similarity techniques based on hashing/projection, this
is too small~\cite{InMo98,AnIn08,ShLi14,SL15,SL15-2,AIL+15}.
Techniques based on LSH and projection usually detect similarities
above 0.8 or higher. Mathematically, these methods have storage
complexities that scale as $1/\tau^2$, and are simply infeasible
when $\tau$ is (say) $0.2$.

We stress that this point does not receive much attention. But in our
view, it is the primary bottleneck behind the lack of methods to solve
\Prb{hemp} for many real applications.

\textbf{The practical challenge:} This leads us to main impetus
behind our work.

\emph{For the matrix $A$ corresponding to the
Twitter network with $O(100\textrm{B})$ edges, find
(as many as possible) entries in $A^TA$ above $0.2$.
For a majority of users, reliably find many $0.2$-similar
users.}

%
%
%
% A common representation of objects is vectors in $\RR^d$, for some
% large dimension $d$. Given two vectors $\vec{v}$ and $\vec{w}$, a common measure of similarity
% is their dot-product $\vec{v} \cdot \vec{w}$. This formulation subsumes are variety
% of settings for a variety of applications: link prediction by finding similar nodes~\cite{AdAd03,LiKl07},
% product recommendation~\cite{CrKoTu10}, text analysis through similarity~\cite{SaAlBu93,BeDuOb95}, data join operations~\cite{AnPi05}, dedup operations for
% data cleaning~\cite{KaMeCh05}. \Sesh{More cites. Should reword, because it's too close to my old paper.}
%
% Often, given collections of such objects, we wish to find all pairs of objects above
% some similarity threshold. In many of the above applications, the vectors are non-negative (or even
% binary),
% since entries represent magnitude of interaction or frequency counts.
%
% We are typically interested in \Prb{hemp} when the vectors all have the same (or similar)
% norm. In this case, this exactly becomes the problem of finding pairs with high cosine
% similarity.
%
% We formulated \Prb{hemp} because of the following data analysis challenge
% from Twitter. Represented every user by the vector of followers in Twitter.
% Determine all pairs of users whose cosine similarity is at least $0.1$.
% In matrix terms: \emph{given a $10^9 \times 10^9$ non-negative matrix $A$ with
% 100 billion non-zeroes, determine all entries in $A^TA$ above $0.1$.}
%
\subsection{Why previous approaches fail} \label{sec:prev}

The challenge described above exemplifies where big data forces an
algorithmic rethink. Matrix multiplication and variants thereof have
been well-studied in the literature, but no solution works for such a
large matrix.  If a matrix $A$ has a 100 billion non-zeroes, it takes
upwards of 1TB just to store the entries. This is more than an order
of magnitude of the storage of a commodity machine in a cluster. Any
approach of partitioning $A$ into submatrices cannot scale.

There are highly tuned libraries like
Intel MKL's BLAS~\cite{MKL14} and CSparse~\cite{Davis06}.
But any sparse matrix multiplication routine~\cite{Gu78,AmPa09} will generate all triples
$(i,i',j)$ such that $A_{i,j}A_{i',j} \neq 0$. In our example, this
turns out to be more than 100 \emph{trillion} triples. This is infeasible
even for a large industrial-strength cluster.

Starting from the work of Drineas, Kannan, and Mahoney, there is rich
line of results on approximate matrix multiplication by subsampling
rows of the
matrix~\cite{DrKa01,DrMa05,DrKaMa06,Sa06,BeWo08,MaZo11,Pa13,HoIp15}.
These methods generate approximate products according to Frobenius
norm using outer products of columns. This would result in dense
matrices, which is clearly infeasible at our scale.  In any case, the
large entries (of interest) contribute to a small part of the output.

\textbf{Why communication matters:} There are upper bounds on the
total communication even in industrial-strength Hadoop clusters, and
in this work we consider our upper bound to be about
100TB\footnote{Note that if a reducer were to process 5GB of data
  each, processing 100TB would require 20,000 reducers.}. A promising
approach for \Prb{hemp} is the wedge sampling method of
Cohen-Lewis~\cite{CoLe99}, which was further developed in the diamond
sampling work of Ballard et al~\cite{BaKo+15}.  The idea is to set up
a linear-sized data structure that can sample indices of entries
proportional to value (or values squared in~\cite{BaKo+15}). One then
generates many samples, and picks the index pairs that occur most
frequently. These samples can be generated in a distributed manner, as
shown by Zadeh and Goel~\cite{ZaGo13}.

The problem is in the final communication. The sampling calculations show
that about $10\tau^{-1} \sum_{i,j} \vec{a_i}\cdot\vec{a_j}$ samples are required
to get all entries above $\tau$ with high probability. These samples must be collected/shuffled
to actually find the large entries. In our setting, this is upwards of 1000TB
of communication.

\textbf{Locality Sensitive Hashing:} In the normalized setting,
maximizing dot product is equivalent to minimizing distance. Thus,
\Prb{hemp} can be cast in terms of finding all pairs of points within
some distance threshold. A powerful technique for this problem is
Locality Sensitive Hashing (LSH)~\cite{InMo98,GiInMo99,AnIn08}.
Recent results by Shrivastava and Li use LSH ideas for the MIPS
problem~\cite{ShLi14,SL15,SL15-2}.  This essentially involves
carefully chosen low dimensional projections with a reverse index for
fast lookup. It is well-known that LSH requires building hashes that
are a few orders of magnitude more than the data size. Furthermore, in
our setting, we need to make hundreds of millions of queries, which
involve constructing all the hashes, and shuffling them to find the
near neighbors.  Again, this hits the communication bottleneck.

\subsection{Results} \label{sec:results}

We design \whimp, a distributed algorithm to solve \Prb{hemp}. We specifically describe
and implement \whimp{} in the MapReduce model, since it is the most appropriate for our applications.

\textbf{Theoretical analysis:} \whimp{} is a novel combination of
wedge sampling ideas from Cohen-Lewis~\cite{CoLe99} with random
projection-based hashes first described by Charikar~\cite{Ch02}.  We
give a detailed theoretical analysis of \whimp{} and prove that it has
near optimal communication/shuffle cost, with a computation cost
comparable to the state-of-the-art. To the best of our knowledge, it
is the first algorithm to have such strong guarentees on the
communication cost.  \whimp{} has a provable precision and recall
guarantee, in that it outputs all large entries, and does not output
small entries.

\textbf{Empirical demonstration:} We implement \whimp{} on Hadoop and
test it on a collection of large networks. Our largest network is {\tt
  flock}, the Twitter network with tens of billions of non-zeroes. We
present results in~\Fig{intro}.  For evaluation, we compute ground
truth for a stratified sample of users (details in~\Sec{setup}). All
empirical results are with respect to this evaluation.  Observe the
high quality of precision and recall for $\tau = 0.2$.  For all
instances other than {\tt flock} (all have non-zeros between 1B to
100B), the accuracy is near perfect. For {\tt flock}, \whimp{}
dominates a precision-recall over $(0.7, 0.7)$, a significant advance
for \Prb{hemp} at this scale.

Even more impressive are the distribution of precision-recall values. For each
user in the evaluation sample (and for a specific setting of parameters in \whimp),
we compute the precision and recall $0.2$-similar vectors. We plot the cumulative
histogram of the minimum of the precision and recall (a lower bound on any F-score)
for two of the largest datasets, {\tt eu} (a web network) and {\tt flock}.
For more than 75\% of the users, we get a precision and recall of more than $0.7$
(for {\tt eu}, the results are even better). Thus, we are able to meet our challenge
of getting accurate results on an overwhelming majority of users.
(We note that in recent advances in using hashing techniques~\cite{ShLi14,SL15,SL15-2},
precision-recall curves rarely dominate the point $(0.4, 0.4)$.)

%
% Our main success with a matrix of more than a 100 billion non-zeroes. We are able to produce
% results on this matrix, and show that it produces excellent results for famous celebrities
% (one of the main industrial applications). We validate by computing brute-force result
% for a small set of randomly chosen users. We also test \whimp{} on a variety of public datasets,
% all above a billion non-zeroes. In all cases, \whimp{} successfully finds many large entries.
% To the best of our knowledge, there was no previous algorithm that could achieve such results in practice.
%

\section{Problem Formulation} \label{sec:prelims} Recall that the
problem of finding similar users is a special case of
Problem~\ref{prb:hemp}. Since our results extend to the more general
setting, in our presentation we focus on the $A^TB$ formulation for
given matrices $A$ and $B$. The set of columns of $A$ is the index set
$[m]$, denoted by $C_A$. Similarly, the set of columns of $B$, indexed by
$[n]$, is denoted by $C_B$. The dimensions of the underlying space are
indexed by $D = [d]$.  We use $a_1, \ldots$ to denote columns of $A$,
$b_1, \ldots$ for columns in $B$, and $r_1, r_2, \ldots$ for
dimensions. For convenience, we assume wlog that $n \geq m$.

We denote rows and columns of $A$ by $A_{d,*}$ and $A_{*,a}$
respectively. And similar notation is used for $B$.  We also use
$\nnz(\cdot)$ to denote the number of non-zeroes in a matrix.  For any
matrix $M$ and $\sigma \in \RR$, the thresholded matrix
$[M]_{\geq \sigma}$ keeps all values in $M$ that are at least
$\sigma$. In other words, $([M]_{\geq \sigma})_{i,j} = M_{i,j}$ if
$M_{i,j} \geq \sigma$ and zero otherwise. We use $\|M\|_1$ to be the
entrywise 1-norm. We will assume that $\|A^TB\|_1 \geq 1$. This is a
minor technical assumption, and one that always holds in matrix
products of interest.

We can naturally represent $A$ as a (weighted) bipartite graph
$G_A = (C_A, D, E_A)$, where an edge $(a,d)$ is present iff
$A_{d,a} \neq 0$. Analogously, we can define the bipartite graph
$G_B$. Their union $G_A \cup G_B$ is a tripartite graph denoted by
$G_{A,B}$. For any vertex $v$ in $G_{A,B}$, we use $N(v)$ for the
neighborhood of $v$.

Finally, we will assume the existence of a Gaussian random number
generator $g$. Given a binary string $x$ as input,
$g(x) \sim \cN(0,1)$. We assume that all values of $g$ are
independent.

\textbf{The computational model:} While our implementation (and focus)
is on MapReduce, it is convenient to think of an abstract distributed
computational model that is also a close proxy for MapReduce in our
setting~\cite{GoelMunagala2012}. This allows for a transparent
explanation of the computation and communication cost.

Let each vertex in $G_{A,B}$ be associated with a different
processor. Communication only occurs along edges of $G_{A,B}$, and
occurs synchronously. Each \emph{round} of communication involves a
single communication over all edges of $G_{A,B}$.

\section{High Level Description} \label{sec:high}

The starting point of our \whimp{} algorithm (Wedges and Hashes in
Matrix Product) is the wedge sampling method of Cohen-Lewis. A
distributed MapReduce implementation of wedge sampling (for the
special case of $A = B$) was given by Zadeh-Goel~\cite{ZaGo13}. In
effect, the main distributed step in wedge sampling is the
following. For each dimension $r \in [d]$ (independently and in
parallel), we construct two distributions on the index sets of vectors
in $A$ and $B$. We then choose a set of independent samples for each
of these distributions, to get pairs $(a,b)$, where $a$ indexes a
vector in $A$, and $b$ indexes a vector in $B$. These are the
``candidates" for high similarity. \emph{If enough candidates are
  generated}, we are guaranteed that the candidates that occur with
high enough frequency are exactly the large entries of $A^TB$.

The primary bottleneck with this approach is that the vast majority of
pairs generated occur infrequently, but dominate the total shuffle
cost. In particular, most non-zero entries in $A^TB$ are very small,
but in total, these entries account for most of $\|A^TB\|_1$. Thus,
these low similarity value pairs dominate the output of wedge sampling.

Our main idea is to construct an efficient, local ``approximate
oracle" for deciding if $A_{*,a}\cdot B_{*,b} \geq \tau$. This is
achieved by adapting the well-known SimHash projection scheme of
Charikar~\cite{Ch02}. For every vector $\vec{v}$ in our input, we
construct a compact logarithmic sized hash $h(\vec{v})$. By the
properties of SimHash, it is (approximately) possible to determine if
$\vec{u} \cdot \vec{v} \geq \tau$ only given the hashes $h(\vec{u})$
and $h(\vec{v})$. These hashes can be constructed by random
projections using near-linear communication. Now, each machine that
processes dimension $r$ (of the wedge sampling algorithm) collects
every hash $h(A_{*,a})$ for each $a$ such that $A_{r,a} \neq 0$
(similarly for $B$). This adds an extra near-linear communication
step, but all these hashes can now be stored locally in the machine
computing wedge samples for dimension $r$. This machines runs the same
wedge sampling procedure as before, but now when it generates a
candidate $(a,b)$, it first checks if $A_{*,a}\cdot B_{*,b} \geq \tau$
using the SimHash oracle. And this pair is emitted iff this condition
passes. Thus, the communication of this step is just the desired
output, since very few low similarity pairs are emitted.  The total
CPU/computation cost remains the same as the Cohen-Lewis algorithm.

\section{The significance of the main theorem} \label{sec:sign}

Before describing the actual algorithm, we state the main theorem and
briefly describe its significance.
\begin{theorem} \label{thm:weed}

  Given input matrices $A, B$ and threshold $\tau$, denote the set of
  index pairs output by \whimp{} algorithm by $S$. Then, fixing
  parameters $\ell = \lceil c\tau^{-2}\log n\rceil$,
  $s = (c(\log n)/\tau)$, and $\sigma = \tau/2$ for a sufficiently
  large constant $c$, the \whimp{} algorithm has the following
  properties with probability at least $1-1/n^2$:
\begin{asparaitem}
    \item \textrm{[Recall:]} If $(A^TB)_{a,b} \geq \tau$, $(a,b)$ is output.
    \item \textrm{[Precision:]} If $(a,b)$ is output, $(A^TB)_{a,b} \geq \tau/4$.
    \item The total computation cost is $O(\tau^{-1}\|A^TB\|_1\log n $
    $+ \tau^{-2}$ $(\nnz(A) $ $+ \nnz(B))\log n)$.
    \item The total communication cost is $O((\tau^{-1}\log n)\|[A^TB]_{\geq \tau/4}\|_1  $ $+ \nnz(A) + \nnz(B) +\tau^{-2}(m+n)\log n)$.
\end{asparaitem}

\end{theorem}

As labeled above, the first two items above are recall and precision.
The first term in the total computation cost is exactly that of
vanilla wedge sampling, $\tau^{-1}\|A^TB\|_1\log n$, while the second
is an extra near-linear term.  The total communication of wedge
sampling is also $\tau^{-1}\|A^TB\|_1\log n$.  Note that \whimp{} has
a communication of \\
$\tau^{-1}\|[A^TB]_{\geq \tau/4}\|_1\log n$.  Since
all entries in $A^TB$ are at most $1$,
$\|[A^TB]_{\geq \tau/4}\|_1 \leq \nnz([A^TB]_{\geq \tau/4})$.  Thus,
the communication of \whimp{} is at most
$(\tau^{-1}\log n) \nnz([A^TB]_{\geq \tau/4}$ plus an additional
linear term. The former is (up to the $\tau^{-1}\log n$ term) simply
the size of the output, and must be paid by any algorithm that outputs
all entries above $\tau/4$. Finally, we emphasize that the constant of
$4$ is merely a matter of convenience, and can be replaced with any
constant $(1+\delta)$.

In summary, \Thm{weed} asserts that \whimp{} has (barring additional
near-linear terms) the same computation cost as wedge sampling, with
nearly optimal communication cost.

\section{The \whimp{} algorithm} \label{sec:whimp}

The \whimp{} algorithm goes through three rounds of communication, each
of which are described in detail in Figure~\ref{alg:whimp}. The output
of \whimp{} is a list of triples $((a,b),\est_{a,b})$, where
$\est_{a,b}$ is an estimate for $(A^TB)_{a,b}$.  Abusing notation, we
say a pair $(a,b)$ is output, if it is part of some triple that is
output.

In each round, we have a step ``Gather". The last round has an output operation.
These are the communication operation. All other
steps are compute operations that are local to the processor involved.

\begin{figure}[!hbtp]
\fbox{
\begin{minipage}{0.45\textwidth}
{\whimp{} Round 1 (Hash Computation):

\smallskip
\begin{compactenum}
    \item For each $a \in C_A$:
    \begin{compactenum}
        \item Gather column $A_{*,a}$.
        \item Compute $\|A_{*,a}\|_2$.
        \item \label{step:ell} Compute bit array $h_a$ of length $\ell$
        as follows: $h_a[i] = \sgn\left(\sum_{r \in [d]} g(\langle r,i\rangle) A_{r,a}\right)$.
    \end{compactenum}
    \item Perform all the above operations for all $b \in C_B$.
\end{compactenum}
}
\end{minipage}}
\fbox{
\begin{minipage}{0.45\textwidth}
{\whimp{} Round 2 (Weight Computation):

\smallskip
\begin{compactenum}
    \item For all $r \in [d]$:
    \begin{compactenum}
        \item Gather rows $A_{r,*}$ and $B_{r,*}$.
        \item Compute $\|A_{r,*}\|_1$ and construct a data structure that samples $a \in C_A$ proportional
        to $A_{r,a}/\|A_{r,*}\|_1$. Call this distribution $\cA_r$.
        \item Similarly compute $\|B_{r,*}\|_1$ and sampling data structure for $\cB_r$.
    \end{compactenum}
\end{compactenum}
}
\end{minipage}}
\fbox{
\begin{minipage}{0.45\textwidth}
{\whimp{} Round 3 (Candidate Generation):

\smallskip
\begin{compactenum}
    \item For all $r \in [d]$:
    \begin{compactenum}
        \item \label{step:sketch} Gather: For all $a,b \in N(r)$, $h_a, h_b, \|A_{*,a}\|_2, \|B_{*,b}\|_2$.
        \item \label{step:s} Repeat $s\|A_{r,*}\|_1 \|B_{r,*}\|_1/$ ($s$ set to $c(\log n)/\tau$) times:
        \begin{compactenum}
            \item \label{step1} Generate $a \sim \cA_r$.
            \item \label{step2} Generate $b \sim \cB_r$.
            \item Denote the Hamming distance between bit arrays $h_a$ and $h_b$ by $\Delta$.
            \item Compute $\est_{a,b} = \|A_{*,a}\|_2 \|B_{*,b}\|_2 \cos(\pi \Delta/\ell)$.
            \item \label{step:sigma} If $\est \geq \sigma$, emit $((a,b),\est_{a,b})$.
        \end{compactenum}

    \end{compactenum}
%     \item \label{step:output} Dedup the tuples emitted in \Step{sigma} to generate output.
\end{compactenum}
}
\end{minipage}}

  \caption{The \whimp{} (Wedges And Hashes In Matrix Product) algorithm}
  \label{alg:whimp}
\end{figure}

\begin{lemma} \label{lem:est-entry} With probability at least $1-1/n^6$
over the randomness of \whimp, for all pairs $(a,b)$, $|\est_{a,b} - A_{*,a}\cdot B_{*,b}| \leq \tau/4$.
\end{lemma}

\begin{proof} First fix a pair $(a,b)$. We have $\est_{a,b} = \|A_{*,a}\|_2 \|B_{*,b}\|_2 $ $\cos(\pi \Delta/\ell)$,
where $\Delta$ is the Hamming distance between $h_a$ and $h_b$.
Note that $h_a[i] = \sgn(\sum_{r \in [d]} g(\langle r,i\rangle) A_{r,a})$.
Let $\vec{v}$ be the $d$-dimension unit vector with $r$th entry proportional
to $g(\langle r,i\rangle)$. Thus, the $r$th component is a random (scaled) Gaussian, and $\vec{v}$
is a uniform (Gaussian) random vector in the unit sphere. We can write $h_a[i] = \sgn(\vec{v}\cdot A_{*,a})$
and $h_b[i] = \sgn(\vec{v}\cdot B_{*,b})$. The probability that $h_a[i] \neq h_b[i]$ is exactly
the probability that the vectors $A_{*,a}$ and $B_{*,b}$ are on different sides of a randomly
chosen hyperplane. By a standard geometric argument~\cite{Ch02},
if $\theta_{a,b}$ is the angle between the vectors $A_{*,a}$ and $B_{*,b}$,
then this probability is $\theta_{a,b}/\pi$.

Define $X_i$ to be the indicator random variable for $h_a[i] \neq h_b[i]$. Note that the Hamming distance
$\Delta = \sum_{i \leq \ell} X_i$ and $\EX[\Delta] = \ell \theta_{a,b}/\pi$. Applying Hoeffding's inequality,
\begin{eqnarray*}& & \Pr[|\Delta - \EX[\Delta]| \geq \ell\tau/(4\pi\|A_{*,a}\|_2 \|B_{*,b}\|_2)] \\
& < & \exp[-(\ell^2 \tau^2/16\pi^2\|A_{*,a}\|^2_2 \|B_{*,b}\|^2_2)/2\ell] \\
& = & \exp(-(c/\tau^2)(\log n)\tau^2/(32\pi^2\|A_{*,a}\|^2_2 \|B_{*,b}\|^2_2)) < n^{-8}
\end{eqnarray*}

Thus, with probability $>1-n^{-8}$,\\
$|\pi\Delta/\ell - \theta_{a,b}| \leq $
$\tau/(4\|A_{*,a}\|_2 \|B_{*,b}\|_2)$.  By the Mean Value Theorem,
$|\cos(\pi\Delta/\ell) - \cos(\theta_{a,b})| \leq \tau/(4\|A_{*,a}\|_2
\|B_{*,b}\|_2)$.  Multiplying by $\|A_{*,a}\|_2 \|B_{*,b}\|_2$, we get
$|\est_{a,b} - A_{*,a}\cdot B_{*,b}| \leq \tau/4$. We take a union
bound over all $\Theta(mn)$ pairs $(a,b)$ to complete the proof.
\end{proof}

We denote a pair $(a,b)$ as \emph{generated} if it is generated in Steps~\ref{step1} and~\ref{step2} during
some iteration. Note that such a pair is actually output iff $\est_{a,b}$ is sufficiently large.

\begin{lemma} \label{lem:gen} With probability at least $1-1/n^3$ over the randomness
of \whimp, the following hold. The total number of triples output is $O((\tau^{-1}\log n) \max(\|[A^T_B]_{\geq \tau/4}\|_1, 1))$.
Furthermore, if $A_{*,a}\cdot B_{*,b} \geq \tau$, $(a,b)$ is output.
\end{lemma}

\begin{proof} Let $X_{a,b,r,i}$ be the indicator random variable for $(a,b)$ being output
in the $i$ iteration for dimension $r$. The total number of times that $(a,b)$ is output
is exactly $X_{a,b} = \sum_{r,i} X_{a,b,r,i}$. By the definition of the distributions $\cA_r$ and $\cB_r$,
$\EX[X_{a,b,r,i}] = \frac{A_{r,a}}{\|A_{r,*}\|_1} \cdot \frac{B_{r,b}}{\|B_{r,*}\|_1}$.
Denote $c(\log n)\|A_{r,*}\|_1 \|B_{r,*}\|_1/\tau$ by $k_r$, the number of samples at dimension $r$.
By linearity of expectation,
\begin{eqnarray*} \EX[X_{a,b}] & = & \sum_{r \leq d} \sum_{i \leq k_r} \frac{A_{r,a}B_{r,b}}{\|A_{r,*}\|_1 \|B_{r,*}\|_1}\\
& = & \sum_{r \leq d} \frac{c(\log n)\|A_{r,*}\|_1 \|B_{r,*}\|_1}{\tau}\cdot \frac{A_{r,a}B_{r,b}}{\|A_{r,*}\|_1 \|B_{r,*}\|_1} \\
& = & c\tau^{-1} \log n \sum_{r \leq d} A_{r,a} B_{r,b} = cA_{*,a}\cdot B_{*,b} \tau^{-1}\log n
\end{eqnarray*}
Note that the random choices in creating the hashes is independent of those generating
the candidates.
By \Lem{est-entry}, with probability $>1-n^{-6}$, the following event (call it $\cE$) holds:
$\forall (a,b)$, $|\est_{a,b} - A_{*,a}\cdot B_{*,b}| \leq \tau/4$. Conditioned on $\cE$,
if $A_{*,a}\cdot B_{*,b} < \tau/4$, then $\est_{a,b} < \tau/2$ and $(a,b)$ is not output.
Let $S = \{(a,b) | A_{*,a} \cdot B_{*,b} \geq \tau/4 \}$.
Let the number of triples output be $Y$.
Conditioned in $\cE$, $Y \geq \sum_{(a,b) \in S} [X_{a,b}]$.
Denote the latter random variable as $Z$.
By linearity of expectation and independence of $X_{a,b}$ from $\cE$,
\begin{eqnarray*}
\EX_\cE[Z] & = & \sum_{(a,b) \in S}\EX_\cE[X_{a,b}] \\
&= & c\tau^{-1}\log n \sum_{(a,b) \in S} A_{*,a} \cdot B_{*,b} \\
&= & c\tau^{-1}\log n \|[A^TB]_{\geq \tau/4}\|_1
\end{eqnarray*}
Furthermore, $Z$ is the sum of Bernoulli random variables.
Thus, we can apply a standard upper-Chernoff bound to the sum above, and deduce
that
\begin{eqnarray*}& & \Pr_\cE[Z \geq 4c\tau^{-1}\log n\max(\|[A^TB]_{\geq \tau/4}\|_1,1)] \\
& \leq & \exp(-4c\tau^{-1}\log n) \leq n^{-10}
\end{eqnarray*}
Thus, conditioned on $\cE$, the probability that $Y$ is greater than
$4c\tau^{-1}\log n$ $\max(\|[A^TB]_{\geq \tau/4}\|_1,1)$ is at most $n^{-10}$.
Since $\Pr[\cE] \leq n^{-6}$, with probability at least $1 - n^{-5}$,
the number of triples output is $O((\tau^{-1}\log n) \max(\|[A^T_B]_{\geq \tau/4}\|_1, 1))$.
This proves the first part.

For the second part now. Fix a pair $(a,b)$ such that $A_{*,a} \cdot B_{*,b} \geq \tau$.
We have $\EX[X_{a,b}] \geq c\log n$. By a standard lower tail Chernoff bound, $\Pr[X_{a,b} \geq (c/2)\log n] \leq n^{-10}$.
Thus, $(a,b)$ is guaranteed to be generated. If event $\cE$ happens, then $\est_{a,b} \geq 3\tau/4$.
By a union bound over the complement events, with probability at least $1-n^{-5}$, $(a,b)$ will be generated
and output. We complete the proof by taking a union bound over all $mn$ pairs $(a,b)$.
% By Cauchy-Schwartz, $A_{*,a}\cdot B_{*,b} \leq 1$. Thus, $\EX[X] \leq 10\tau^{-1}\log n$.
% By an upper-tail Chernoff bound, $\Pr[X \geq 2ec\tau^{-1}\log n] \leq 2^{2ec\tau^{-1}\log n} \leq n^{-7}$.
% If $A_{*,a}\cdot B_{*,b} \geq \tau$, then $\EX[X] \geq c\log n$. By a lower-tail Chernoff bound,
% $\Pr[X \geq (c\log n)/2] \leq \exp(-c\log n/12) \leq n^{-7}$. A union bound over these completes the proof.
\end{proof}

The first two statements of \Thm{weed} hold by \Lem{est-entry} and
\Lem{gen}, and the remaining two statements follow by a
straightforward calculation. Hence we skip the remainder of the proof.

%\begin{proof} (of \Thm{weed}) The first two statements hold by \Lem{est-entry} and \Lem{gen}.
%
%Now for the cost analysis. The computation cost in Round 1 is $O(\nnz(A) + \nnz(B) + \tau^{-2}(m+n)\log n)$.
%In Round 2, we build a Walker's alias table to sample from the distributions. The computation cost is $O(\nnz(A) + \nnz(B)$.
%In Round 3, each samples takes $O(1)$ to generate. The Hamming distance computation takes $O(\tau^{-2}\log n)$ operations.
%Up to constant factors, the computation cost is $\tau^{-2}\log n$ $ \sum_{r \leq d} \|A_{r,*}\|_1 \|B_{r,*}\|_1$.
%%
%\begin{eqnarray*} \sum_{r \leq d} \|A_{r,*}\|_1 \|B_{r,*}\|_1 & = & \sum_{r \leq d} \sum_{a \leq m} A_{r,a} \sum_{b \leq n} B_{r,b}\\
%& = & \sum_{a \leq m} \sum_{b \leq n} \sum_{r \leq d} A_{r,a} B_{r,b} \\
%& = & \sum_{a \leq m} \sum_{b \leq n}A_{*,a} \cdot B_{*,b} = \|A^TB\|_1
%\end{eqnarray*}
%%
%
%The communication costs of Rounds 1 and 2 are just $O(\nnz(A) + \nnz(B))$. In Round 3, gathering the hashes and the norms
%costs $O(\nnz(A) + \nnz(B) + \tau^{-2}(m+n)\log n)$.
%The bound on the output is given in \Lem{gen}. Note that we can drop the maximum term. If $\|[A^TB]_{\geq \tau/4}\|_1 < 1$,
%then the communication cost is subsumed by $\tau^{-2}(m+n)\log n$.
%\end{proof}

\section{Implementing \whimp} \label{sec:implement}

We implement and deploy \whimp{} in Hadoop~\cite{SKRC10}, which is an
open source implementation of MapReduce~\cite{DeanGhemawat08}. Our
experiments were run on Twitter's production Hadoop cluster, aspects
of which have been described before in~\cite{LR13, LLLLR12,
  GSWY13}. In this section, we discuss our \whimp{} parameter choices
and some engineering details. As explained earlier, all our
experiments have $A = B$.

It is helpful to discuss the quality measures.
Suppose we wish to find all entries above some threshold $\tau > 0$.
Typical choices are in the range $[0.1,0.5]$ (cosine values are rarely
higher in our applications). The support of $[A^TA]_{\geq \tau}$
is denoted by $H_\tau$, and this is the set of pairs that we wish to find.
Let the output of \whimp{} be $S$. The natural aim is to maximize both
precision and recall.
\begin{asparaitem}
    \item Precision: the fraction of output that is ``correct", $|H_\tau \cap S|/|S|$.
    \item Recall: the fraction of $H_\tau$ that is output, $|H_\tau \cap S|/|H_\tau|$.
\end{asparaitem}

\medskip

There are three parameter choices in \whimp, as described in \Thm{weed}. We show
practical settings for these parameters.

\begin{figure*}[t]
  \centering
  \subfloat{\label{fig:pr20}
    % TRIM LEFT BOTTOM RIGHT TOP
    \includegraphics[width=.3\textwidth,trim=0 0 0 0]{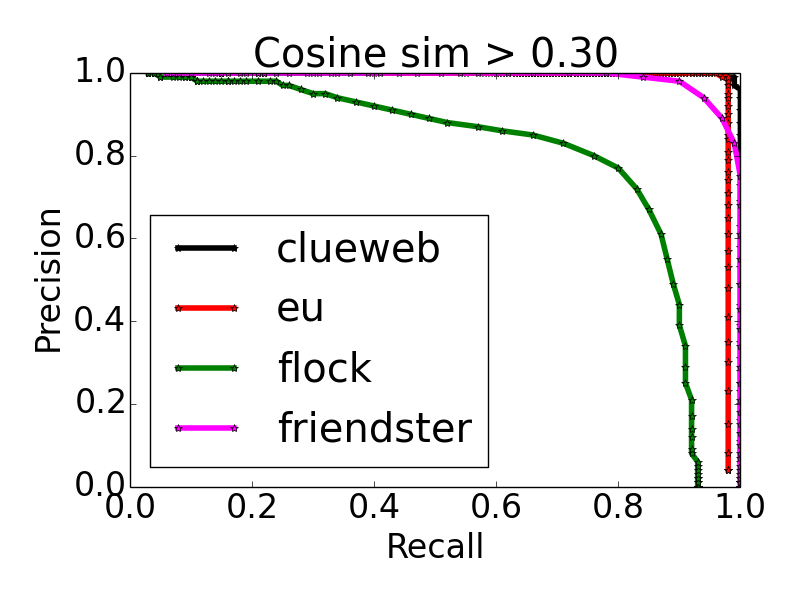}
  }
  \subfloat{\label{fig:pr40}
    % TRIM LEFT BOTTOM RIGHT TOP
    \includegraphics[width=.3\textwidth,trim=0 0 0 0]{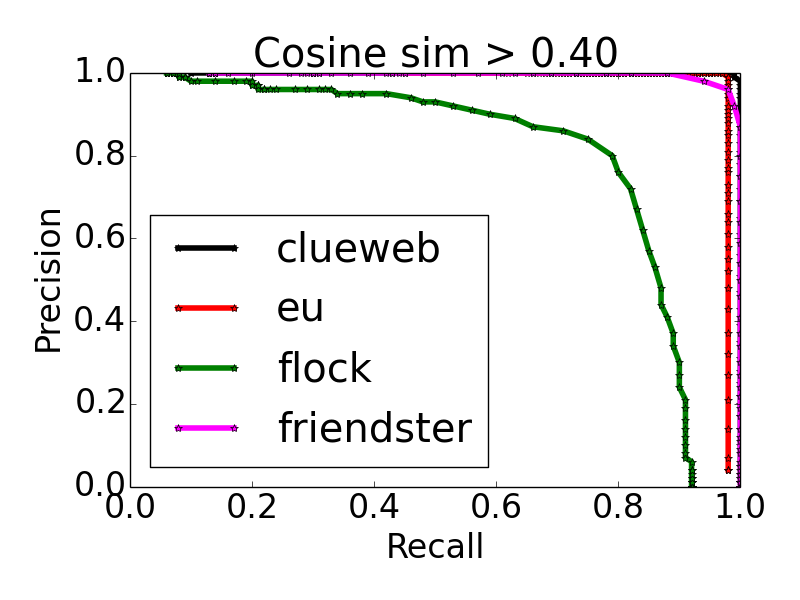}
  }
  \subfloat{\label{fig:pr60}
    % TRIM LEFT BOTTOM RIGHT TOP
    \includegraphics[width=.3\textwidth,trim=0 0 0 0]{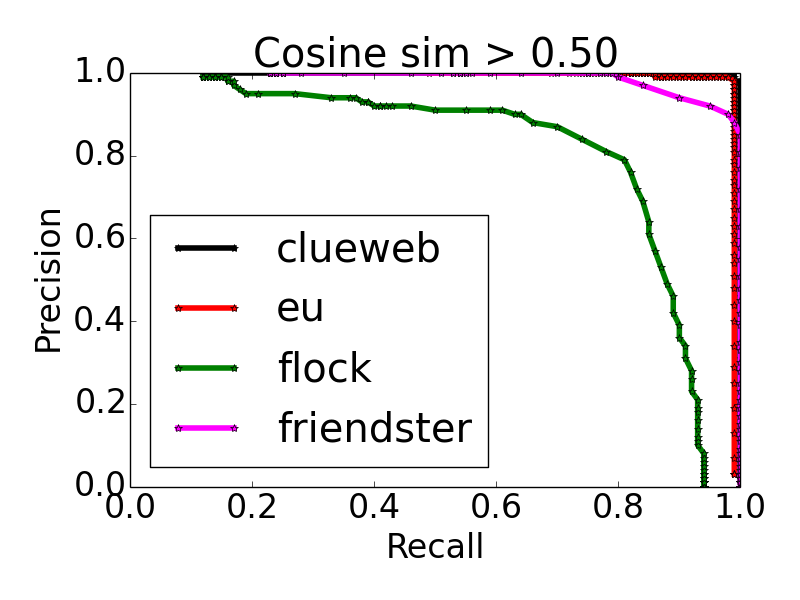}
  }
  \caption{Precision-recall curves}
  \label{fig:pr-all}
\end{figure*}

\textbf{$\ell$, the sketch length:} This appears in \Step{ell}
of Round 1. Larger $\ell$ implies better accuracy
for the SimHash sketch, and thereby leads to higher precision and recall.
On the other hand, the communication in Round 3 requires emitting all sketches,
and thus, it is linear in $\ell$.

A rough rule of thumb is as follows: we wish to distinguish $A_{*,a}\cdot A_{*,b} = 0$
from $A_{*,a} \cdot A_{*,b} > \tau$. (Of course, we wish for more, but this argument
suffices to give reasonable values for $\ell$.) Consider a single bit of SimHash.
In the former case, $\Pr[h(A_{*,a}) = h(A_{*,b}) = 1/2$, while in the latter
case $\Pr[h(A_{*,a}) = h(A_{*,b)}] = 1 - \theta_{a,b}/\pi $
$= \cos^{-1}(A_{*,a}\cdot A_{*,b})/\pi \geq 1 - \cos^{-1}(\tau)/\pi$.
It is convenient to express the latter as $\Pr[h(A_{*,a}) = h(A_{*,b})] \geq 1/2 + \delta$,
where $\delta = 1/2 - \cos^{-1}(\tau)/\pi$.

Standard binomial tail bounds tells us that $1/\delta^2$ independent
SimHash bits are necessary to distinguish the two cases. For
convergence, at least one order of magnitude more samples are
required, so $\ell$ should be around $10/\delta^2$.  Plugging in some
values, for $\tau = 0.1$, $\delta = 0.03$, and $\ell$ should be
$11,000$.  For $\tau = 0.2$, we get $\ell$ to be $2,400$. In general,
the size of $\ell$ is around 1 kilobyte.

\medskip

\textbf{s, the oversampling factor:} This parameter appears
in~\Step{s} of Round 3, and determines the number of wedge samples
generated. The easiest way to think about $s$ is in terms of vanilla
wedge sampling. Going through the calculations, the total number of
wedge samples (over the entire procedure) is exactly
$s\sum_r \|A_{r,*}\|_1\|A_{r,*}\|_1 = s\|A^TA\|_1$.  Fix a pair
$(a,b) \in H_\tau$, with dot product exactly $\tau$.  The probability
that a single wedge sample produces $(a,b)$ is
$A_{*,a} \cdot A_{*,b}/\|A^TA\|_1 = \tau/\|A^TA\|_1$.  Thus, \whimp{}
generates this pair (expected) $\tau/\|A^TA\|_1 \times s\|A^TA\|_1$
$= \tau s$ times.

The more samples we choose, the higher likelihood of finding a pair
$(a,b) \in S_\tau$.  On the other hand, observe that pairs in $H_\tau$
are generated $\tau s$ times, and increasing $s$ increases the
communication in Round 3.  Thus, we require $s$ to be at least
$1/\tau$, and our rule of thumb is $10/\tau$ to get convergence.

\medskip

\textbf{$\sigma$, the filtering value:} This is used in the final
operation, \Step{sigma}, and decides which pairs are actually
output. The effect of $\sigma$ is coupled with the accuracy of the
SimHash sketch. If the SimHash estimate is perfect, then $\sigma$
should just be $\tau$. In practice, we modify $\sigma$ to account for
SimHash error. Higher $\sigma$ imposes a stricter filer and improves
precision at the cost of recall. And the opposite happens for lower
$\sigma$.  In most runs, we simply set $\sigma = \tau$. We vary
$\sigma$ to generate precision-recall curves.
\begin{figure*}[t]
  \centering
  \subfloat{\label{fig:clueweb40}
    % TRIM LEFT BOTTOM RIGHT TOP
    \includegraphics[width=.3\textwidth,trim=0 0 0 0]{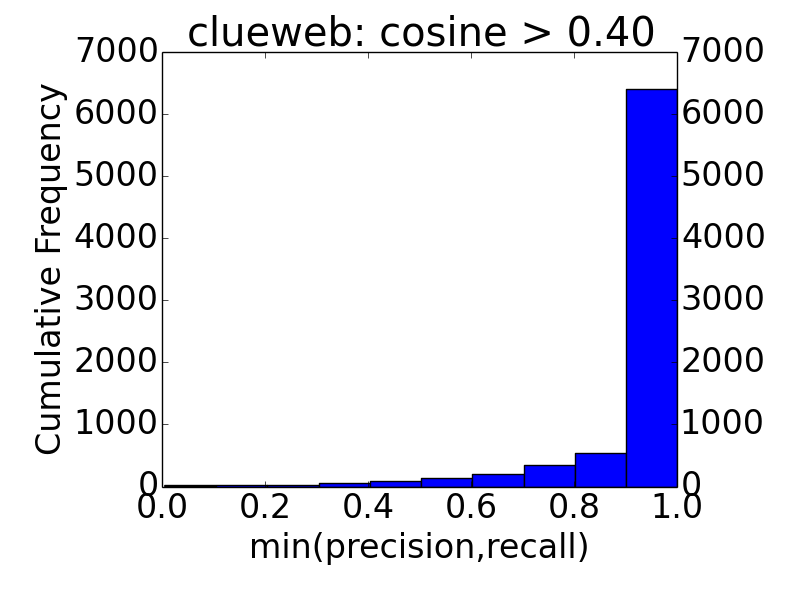}
  }
  \subfloat{\label{fig:eu40}
    % TRIM LEFT BOTTOM RIGHT TOP
    \includegraphics[width=.3\textwidth,trim=0 0 0 0]{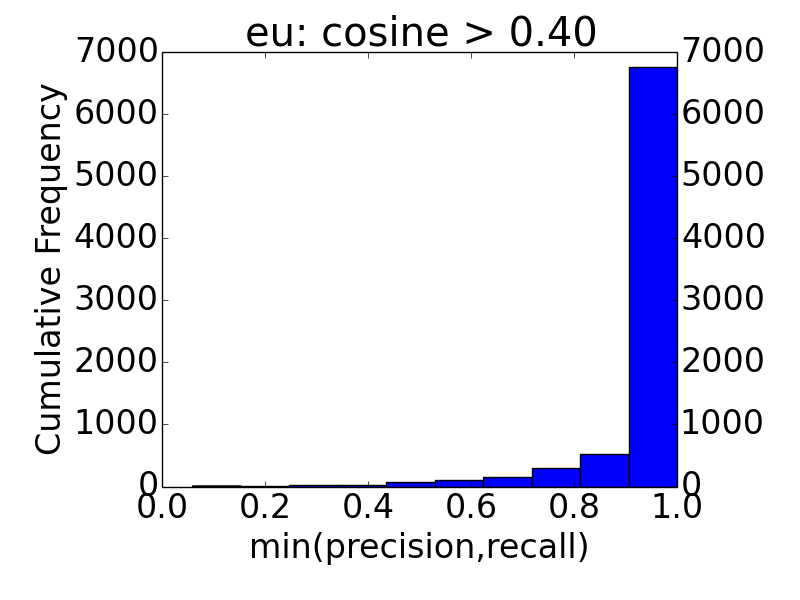}
  }
  \subfloat{\label{fig:flock40}
    % TRIM LEFT BOTTOM RIGHT TOP
    \includegraphics[width=.3\textwidth,trim=0 0 0 0]{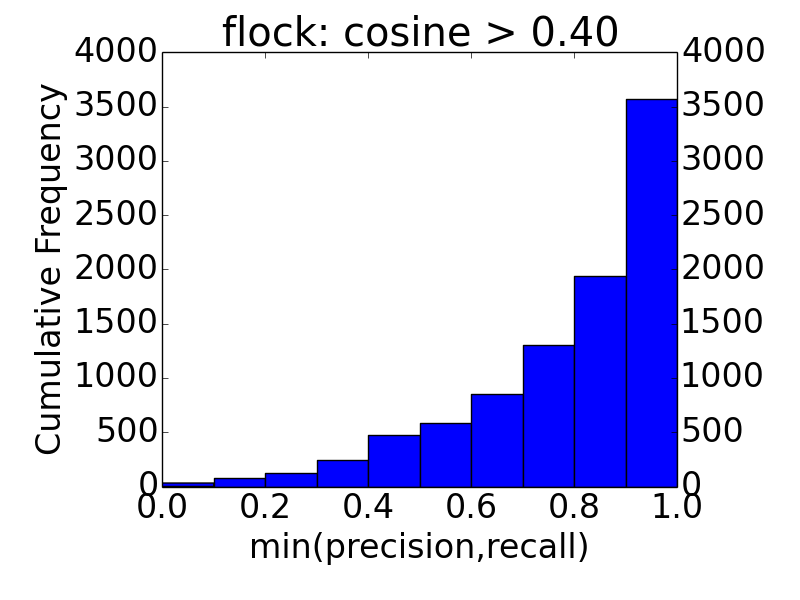}
  }
  \caption{Per-user precision-recall histograms for $\tau = 0.4$}
  \label{fig:peruser40}
\end{figure*}

\vfill\eject
\section{Experimental Setup} \label{sec:setup}

As mentioned earlier, we run all experiments on Twitter's Hadoop
cluster. All the code for this work was written in Scalding, which is
Twitter's Scala API to Cascading, an open-source framework for
building dataflows that can be executed on Hadoop. These are all
mature production systems, aspects of which have been discussed in
detail elsewhere~\cite{LR13, LLLLR12, GSWY13}.

\textbf{Datasets:} We choose four large datasets. Two of them, {\tt
  clueweb} and {\tt eu} are webgraphs. The dataset {\tt friendster} is
a social network, and is available from the Stanford Large Network
Dataset Collection~\cite{snapnets}. The two webgraphs were obtained
from the LAW graph repository~\cite{BRSLLP,BoVWFI}. Apart from these
public datasets, we also report results on our proprietary dataset,
{\tt flock}, which is the Twitter follow graph.

We interpret the graph as vectors in the following way. For each vertex,
we take the incidence vector of the \emph{in-neighborhood}. Thus,
two vertices are similar if they are followed by a similar set of other
vertices. This is an extremely important signal for Twitter's recommendation
system~\cite{GGLS+13}, our main motivating problem. For consistency, we apply
the same viewpoint to all the datasets.

We apply a standard cleaning procedure (for similarity) and remove
high out-degrees.  In other words, if some vertex $v$ has more than
10K followers (outdegree $>$ 10K), we remove all these edges. (We do
not remove the vertex, but rather only its out-edges.) Intuitively,
the fact that two vertices are followed by $v$ is not a useful signal
for similarity. In {\tt flock} and {\tt friendster}, such vertices are
typically spammers and should be ignored.  For webgraphs, a page
linking to more than 10K other pages is probably not useful for
similarity measurement.

\begin{table}[h]
    \begin{tabular}{l || c | c | c}
    Dataset & Dimensions $n=d$ & Size (nnz) & $|A^TA|_1$ \\ \hline \hline
    {\tt friendster} & 65M & 1.6B & 7.2E9 \\ \hline
    {\tt clueweb} &  978M & 42B & 6.8E10 \\ \hline
    {\tt eu} &  1.1B & 84B & 1.9E11 \\ \hline
    {\tt flock} & - & $O(100B)$ & 5.1E12 \\ \hline
    \end{tabular}
    \caption{Details on Datasets}
    \label{tab:data}
\end{table}

We give the size of the datasets in~\Tab{data}. (This is after
cleaning, which removes at most 5\% of the edges. Exact sizes for {\tt
  flock} cannot be revealed but we do report aggregate results where
possible.) Since the underlying matrix $A$ is square, $n=d$. All
instances have at least a billion non-zeros. To give a sense of scale,
the raw storage of 40B non-zeros (as a list of edges/pairs, each of
which is two longs) is roughly half a terrabyte. This is beyond the
memory of most commodity machines or nodes in a small cluster,
underscoring the challenge in designing distributed algorithms.

\textbf{Parameters:} We set the parameters of \whimp{} as follows.
Our focus is typically on $\tau > 0.1$, though we shall present results
for varying $\tau \in [0.1, 0.5]$.
The sketch length $\ell$ is 8192 (1KB sketch size); the oversampling
factor $s$ is $150$; $\sigma$ is just $\tau$. For getting precision-recall
curves, we vary $\sigma$, as discussed in \Sec{implement}.

\textbf{Evaluation:} Computing $A^TA$ exactly is infeasible at these sizes.
A natural evaluation would be pick a random sample of vertices
and determine all similar vertices for each vertex in the sample.
(In terms of matrices, this involves sampling columns of $A$ to get
a thinner matrix $B$, and then computing $A^TB$ explicitly).
Then, we look at the output of \whimp{} and measure the number of similar
pairs (among this sample) it found. An issue with pure uniform sampling
is that most vertices tend to be low degree (the columns have high sparsity).
In recommendation applications, we care for accurate behavior at all scales.

We perform a stratified sampling of columns to generate ground truth.
For integer $i$, we create a bucket with all vertices whose indegree (vector sparsity)
is in the range $[10^i,10^{i+1})$. We then uniformly sample 1000 vertices
from each bucket to get a stratified sample of vertices/columns.
All evaluation is performed with respect to the exact results for this stratified sample.

\section{Experimental results} \label{sec:exp}

\textbf{Precision-recall curves:} We use threshold $\tau$ of $0.2$,
$0.4$, $0.6$.  We compute precision-recall curves for \whimp{} on all
the datasets, and present the results in \Fig{pr-all}. Observe the
high quality results on {\tt clueweb}, {\tt eu}, and {\tt friendster}:
for $\tau\geq 0.4$, the results are near perfect.  The worst behavior
is that of {\tt flock}, which still dominates a precision and recall
of $0.7$ in all cases.  Thus, \whimp{} is near perfect when $\nnz(A)$
has substantially fewer than 100B entries (as our theory
predicts). The extreme size of {\tt flock} probably requires even
larger parameter settings to get near perfect results.

\begin{figure}
    \centering
    \includegraphics[width=0.3\textwidth]{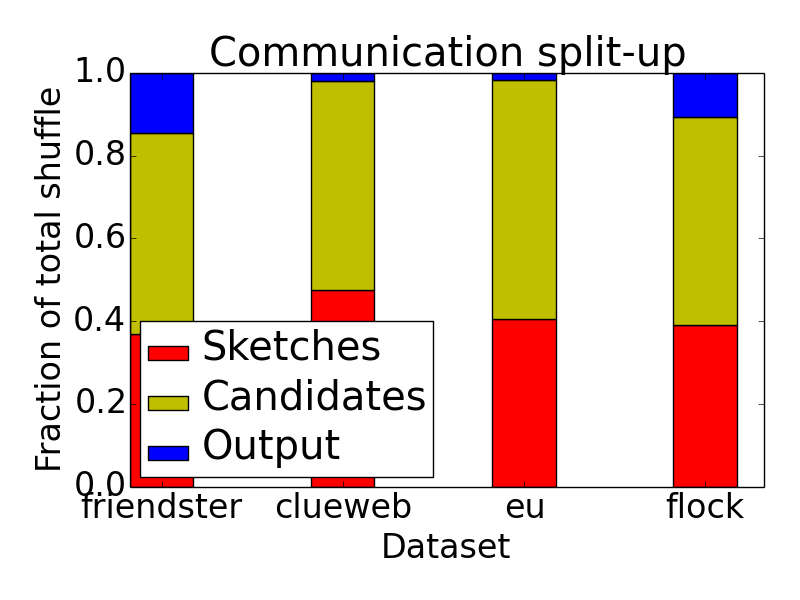}
    \caption{Split-up of shuffle over various rounds for \whimp}
    \label{fig:split}
\end{figure}

\begin{table}[h]
\begin{tabular}{l || c | c | c}
Dataset & \whimp{} (TB)  & DISCO est. (TB) & $\|A^TA\|_1$ \\ \hline \hline
{\tt friendster} & 4.9 & 26.2 & 7.2e+09\\ \hline
{\tt clueweb} & 90.1 & 247.4 & 6.8e+10\\ \hline
{\tt eu} & 225.0 & 691.2 & 1.9e+11\\ \hline
{\tt flock} & 287.0 & 18553.7 & 5.1e+12\\ \hline
\end{tabular}
\caption{Total communication/shuffle cost of \whimp}
\label{tab:shuffle}
\end{table}

\textbf{Per-vertex results:} In recommendation applications,
global precision/recall is less relevant that per-user results.
Can we find similar neighbors for most users, or alternately,
for how many users can we provide accurate results? This is more stringent quality
metric than just the number of entries in $[A^TA]_{\geq \tau}$ obtained.

In the following experiment, we simply set the filtering value $\sigma$ to be $\tau$.
We vary $\tau$ in $0.2$, $0.4$, etc. For each dataset and each vertex in the evaluation sample,
(generation described in \Sec{setup}) we compute the precision and recall for
\whimp{} just for the similar vertices of the sample vertex. We just focus
on the minimum of the precision and recall (this is a lower bound on
any $F_\beta$ score, and is a conservative measure).
The cumulative (over the sample)
histogram of the minimum of the precision and recall is plotted in \Fig{peruser40}.

Just for clarity, we give an equivalent description in terms of matrices.
We compute the (minimum of) precision and recall of entries
above $\tau$ in a specific (sampled) column of $A^TA$. We plot the cumulative histogram
over sampled columns.

For space reasons, we only show the results for $\tau = 0.4$ and
ignore the smallest dataset, {\tt friendster}. The results for {\tt
  clueweb} and {\tt eu} are incredibly accurate: for more than 90\% of
the sample, both precision and recall are above $0.8$, regardless of
$\tau$.  The results for {\tt flock} are extremely good, but not
nearly as accurate.  \whimp{} gets a precision and recall above $0.7$
for at least 75\% of the sample.  We stress the low values of cosine
similarities here: a similarity of $0.2$ is well-below the values
studied in recent LSH-based results~\cite{ShLi14,SL15,SL15-2}. It is
well-known that low similarity values are harder to detect, yet
\whimp{} gets accurate results for an overwhelming majority of the
vertices/users.

\medskip

\textbf{Shuffle cost of \whimp:} The main impetus behind \whimp{} was to get an algorithm with low
shuffle cost. Rounds 1 and 2 only shuffle the input data (and a small factor over it), and
do not pose a bottleneck. Round 3 has two major shuffling steps.
\begin{asparaitem}
    \item Shuffling the sketches: In \Step{sketch}, the sketches are communicated.
    The total cost is the sum of sizes of all sketches, which is $\ell \nnz(A)$.
    \item Shuffling the candidates that are output: In \Step{sigma}, the candidates
    large entries are output. There is an important point here that is irrelevant
    in the theoretical description. We perform a deduplication step to output
    entries only once. This requires a shuffle step after which the final output
    is generated.
\end{asparaitem}
We split communication into three parts: the sketch shuffle, the candidate shuffle,
and the final (deduped) output. The total of all these is presented in \Tab{shuffle}.
(We stress that this is not shuffled together.) The split-up between the various
parts is show in in \Fig{split}. Observe that the sketch and candidate shuffle are
roughly equal. For {\tt friendster} and {\tt flock}, the (deduped) output is itself
more than 10\% of the total shuffle. This (weakly) justifies
the optimality \Thm{weed} in these cases, since
the total communication is at most an order of magnitude more than the desired output.
For the other cases, the output is between 3-5\% of the total shuffle.

\medskip

\textbf{Comparisons with existing art:} No other algorithm works at this scale,
and we were not able to deploy anything else for such large datasets. Nonetheless,
given the parameters of the datasets, we can mathematically argue against other approaches.
\begin{asparaitem}
    \item Wedge sampling of Cohen-Lewis~\cite{CoLe99}, DISCO~\cite{ZaGo13}: Distributed
    version of wedge sampling were given by Zadeh and Goel in their DISCO algorithm~\cite{ZaGo13}.
    But it cannot scale to these sizes. DISCO is equivalent to using Round 2 of \whimp to set up
    weights, and then running Round 3
    without any filtering step (\Step{sigma}). Then, we would look for all pairs $(a,b)$
    that have been emitted sufficiently many times, and make those the final output.
    In this case, CPU and shuffle costs are basically identical,
    since any candidate generated is emitted.

    Consider $(a,b)$ such that $A_{*,a}\cdot A_{*,b} = \tau$. By the wedge sampling calculations,
    $s\|A^TA\|_1$ wedge samples would generate $(a,b)$ an expected $s\tau$ times. We would need
    to ensure that this is concentrated well, since we finally output pairs generated often
    enough. In our experience, setting $s = 50/\tau$ is the bare minimum to get precision/recall
    more than $0.8$. Note that \whimp{} only needs to generate such a wedge sample \emph{once},
    since \Step{sigma} is then guaranteed to output it (assuming SimHash is accurate).
    But vanilla wedge sampling must generate $(a,b)$ with a frequency close to its expectation.
    Thus, \whimp{} can set $s$ closer to (say) $10/\tau$, but this is not enough for the convergence
    of wedge sampling.

    But all the wedges have to be shuffled, and this leads to $10\|A^TA\|_1/\tau$ wedges being
    shuffled. Each wedge is two longs (using standard representations), and that gives
    a ballpark estimate of $160\|A^TA\|_1/\tau$. We definitely care about $\tau = 0.2$,
    and \whimp{} generates results for this setting (\Fig{pr-all}). We compute
    this value for the various datasets in \Tab{shuffle}, and present it as the estimated
    shuffle cost for DISCO.

    Observe that it is significantly large than the \emph{total} shuffle cost of \whimp,
    which is actually split roughly equally into two parts (\Fig{split}).
    The wedge shuffles discussed above are most
    naturally done in a single round. To shuffle more than 200TB would require a more
    complex algorithm that splits the wedge samples into various rounds.
    For {\tt eu} and {\tt flock}, the numbers are more
    than 1000TB, and completely beyond the possibility of engineering. We note that
    {\tt friendster} can probably be handled by the DISCO algorithm.

  \item Locality-Sensitive Hashing~\cite{InMo98,ShLi14}: LSH is an
    important method for nearest neighbor search. Unfortunately, it
    does not perform well when similarities are low but still
    significant (say $\tau = 0.2$). Furthermore, it is well-known to
    require large memory overhead. The basic idea is to hash every
    vector into a ``bucket'' using, say, a small (like 8-bit) SimHash
    sketch. The similarity is explicitly computed on all pairs of
    vectors in a bucket, i.e. those with the same sketch value. This
    process is repeated with a large number of hash functions to
    ensure that most similar pairs are found.

     Using SimHash, the mathematics says roughly the following. (We
     refer the reader to important LSH papers for more
     details~\cite{InMo98,GiInMo99,AnIn08,ShLi14}.)  Let the
     probability of two similar vectors (with cosine similarity above
     $0.2$) having the same SimHash value be denoted $P_1$. Let the
     corresponding probability for two vectors with similarity zero by
     $P_2$. By the SimHash calculations of \Sec{implement},
     $P_1 = 1/2 - \cos^{-1}(0.2)/\pi \approx 0.56$, while $P_2 =
     0.5$. This difference measures the ``gap" obtained by the SimHash
     function. The LSH formula basically tells us that the total
     storage of all the hashes is (at least)
     $n^{1+(\log P_1)/(\log P_2)}$ bytes.  This comes out to be
     $n^{1.83}$. Assuming that $n$ is around 1 billion, the total
     storage is 26K TB. This is astronomically large, and even
     reducing this by a factor of hundred is insufficient for
     feasibility.

\end{asparaitem}

\begin{table}[!t]
\centering
 \small
\caption{Top similar results for a few Twitter accounts, generated
  from \whimp{} on {\tt flock}.}

\begin{tabular}{|c|p{1in}|c|}
\multicolumn{3}{c}{\textbf{Users similar to @www2016ca}} \\ \hline
Rank & Twitter @handle & Score \\ \hline
1 & @WSDMSocial & 0.268 \\
2 & @WWWfirenze & 0.213 \\
3 & @SIGIR2016 & 0.190 \\
4 & @ecir2016 & 0.175 \\
5 & @WSDM2015 & 0.155 \\
\hline
\end{tabular}

\begin{tabular}{|c|p{1in}|c|}
\multicolumn{3}{c}{\textbf{Users similar to @duncanjwatts}} \\ \hline
Rank & Twitter @handle & Score \\ \hline
1 & @ladamic & 0.287 \\
2 & @davidlazer & 0.286 \\
3 & @barabasi & 0.284 \\
4 & @jure & 0.218 \\
5 & @net\_science & 0.200 \\
\hline
\end{tabular}

\begin{tabular}{|c|p{1in}|c|}
\multicolumn{3}{c}{\textbf{Users similar to @POTUS}} \\ \hline
Rank & Twitter @handle & Score \\ \hline
1 & @FLOTUS & 0.387 \\
2 & @HillaryClinton & 0.368 \\
3 & @billclinton & 0.308 \\
4 & @BernieSanders & 0.280 \\
5 & @WhiteHouse & 0.267 \\
\hline
\end{tabular}

\label{table:casestudy}
\end{table}

\textbf{Case Study:} In addition to the demonstration of the
algorithm's performance in terms of raw precision and recall, we also
showcase some examples to illustrate the practical effectiveness of
the approach. Some of these results are presented in
Table~\ref{table:casestudy}. First, note that the cosine score values
that generate the results are around 0.2, which provides justification
for our focus on generating results with these values. Furthermore,
note that even at these values, the results are quite interpretable
and clearly find similar users: for the @www2016ca account, it finds
accounts for other related social networks and data mining
conferences. For @duncanjwatts, who is a network science researcher,
the algorithm finds other network science researchers. And finally, an
example of a very popular user is @POTUS, for whom the algorithm finds
clearly very related accounts.

\bibliographystyle{abbrv}
\bibliography{weed}  % sigproc.bib is the name of the Bibliography in this case

\end{document}